\documentclass{elsarticle}
\usepackage[T1]{fontenc}
\usepackage[latin9]{inputenc}
\usepackage{array}
\usepackage{float}
\usepackage{amsthm}
\usepackage{amsmath}
\usepackage{amssymb}
\usepackage{graphicx}

\makeatletter

\newcommand{\noun}[1]{\textsc{#1}}
\providecommand{\tabularnewline}{\\}

\theoremstyle{plain}
\newtheorem{thm}{\protect\theoremname}[section]
\theoremstyle{plain}
\newtheorem{lem}[thm]{\protect\lemmaname}
\ifx\proof\undefined
\newenvironment{proof}[1][\protect\proofname]{\par
\normalfont\topsep6\p@\@plus6\p@\relax
\trivlist
\itemindent\parindent
\item[\hskip\labelsep
\scshape
#1]\ignorespaces
}{%
\endtrivlist\@endpefalse
}
\providecommand{\proofname}{Proof}
\fi

\usepackage{amsmath}
\usepackage[hidelinks]{hyperref}

\DeclareMathOperator{\PARTS}{PARTS}

\DeclareMathOperator{\adr}{adr}

\@ifundefined{showcaptionsetup}{}{%
 \PassOptionsToPackage{caption=false}{subfig}}
\usepackage{subfig}
\makeatother

\providecommand{\lemmaname}{Lemma}
\providecommand{\theoremname}{Theorem}

\begin{document}
\title{Complexity of a Problem Concerning Reset Words for Eulerian Binary Automata\tnoteref{t1}} 
\tnotetext[t1]{Research supported by the Czech Science Foundation grant GA14-10799S.}

\author{Vojt\v{e}ch Vorel}

\address{Faculty of Mathematics and Physics, Charles University\\Malostranské
nám. 25, 118 00 Prague, Czech Republic}
\begin{abstract}
A word is called a reset word for a deterministic finite automaton
if it maps all the states of the automaton to a unique state. Deciding
about the existence of a reset word of a given maximum length for
a given automaton is known to be an NP-complete problem. We prove
that it remains NP-complete even if restricted to Eulerian automata
with binary alphabets, as it has been conjectured by Martyugin (2011).
\end{abstract}
\maketitle

\section{Introduction and Preliminaries}

A \emph{deterministic finite automaton }is a triple $A=\left(Q,X,\delta\right)$,
where $Q$ and $X$ are finite sets and $\delta$ is an arbitrary
mapping $Q\times X\rightarrow Q$. Elements of $Q$ are called \emph{states},
$X$ is the \emph{alphabet}. The \emph{transition function} $\delta$
can be naturally extended to $Q\times X^{\star}\rightarrow Q$, still
denoted by $\delta$. We extend it also by defining 
\[
\delta\!\left(S,w\right)=\left\{ \delta\!\left(s,w\right)\mid s\in S,w\in X^{\star}\right\} 
\]
 for each $S\subseteq Q$. If the automaton is fixed, we write
\[
r\overset{w}{\longrightarrow}s
\]
instead of $\delta\!\left(r,w\right)=s$.

For a given automaton $A=\left(Q,X,\delta\right)$, we call $w\in X^{\star}$
a \emph{reset word} if 
\[
\left|\delta\!\left(Q,w\right)\right|=1.
\]
If such a word exists, we call the automaton \emph{synchronizing}.
Note that each word having a reset word as a factor is also a reset
word.

A need for finding reset words appears in several fields of mathematics
and engineering. Classical applications (see \citep{VOL1}) include
model-based testing, robotic manipulation, and symbolic dynamics,
but there are important connections also with information theory \citep{TRS1}
and with formal models of biomolecular processes \citep{BON1}.

The \emph{\v{C}ern\'{y} Conjecture}, a longstanding open problem,
claims that each synchronizing automaton has a reset word of length
$\left(\left|Q\right|-1\right)^{2}$. Though it still remains open,
there are many weaker results in this field, see e.g. \citep{STE5,GRE1}
for recent ones%
\footnote{The result published by Trahtman \citep{TRA1} in 2011 has turned
out to be proved incorrectly.%
}.

Various computational problems arise from the study of synchronization:
\begin{itemize}
\item \emph{Given an automaton, decide if it is synchronizing.} Relatively
simple algorithm, which could be traced back to \citep{CER1}, works
in polynomial time. 
\item \emph{Given a synchronizing automaton and a number $d$, decide if
$d$ is the length of shortest reset words.} This has been shown to
be both NP-hard \citep{EPP1} and coNP-hard. More precisely, it is
DP-complete \citep{OLS1}.
\item \emph{Given a synchronizing automaton and a number $d$, decide if
there exists a reset word of length $d$.} This problem is of our
interest. Lying in NP, it is not so computationally hard as the previous
problem. However, it is proven to be NP-complete \citep{EPP1}. Following
the notation of \citep{MAR2}, we call it \noun{Syn}. Assuming that
$\mathcal{M}$ is a class of automata and membership in $\mathcal{M}$
is polynomially decidable, we define a restricted problem:\medskip{}
\\
\begin{tabular}{ll}
\multicolumn{2}{l}{\noun{Syn($\mathcal{M}$)}}\tabularnewline
Input: & synchronizing automaton $A=\left([n],X,\delta\right)\in\mathcal{M}$,
$d\in\mathbb{N}$\tabularnewline
Output: & does $A$ have a reset word of length $d$?\tabularnewline
\end{tabular}
\end{itemize}
An automaton $A=\left(Q,X,\delta\right)$ is \emph{Eulerian }if
\[
\sum_{x\in X}\left|\left\{ r\in Q\mid\delta\!\left(r,x\right)=q\right\} \right|=\left|X\right|
\]
for each $q\in Q$. Informally, there should be exactly $\left|X\right|$
transitions incoming to each state.\emph{ }An automaton is\emph{ binary}
if $\left|X\right|=2$. The classes of Eulerian and binary automata
are denoted by $\mathcal{EU}$ and $\mathcal{AL}_{2}$ respectively.

Previous results about various restrictions of \noun{Syn} can be found
in \citep{EPP1,MAR3,MAR2}. Some of these problems turned out to be
polynomially solvable, others are NP-complete. In \citep{MAR2} Martyugin
conjectured that \noun{Syn($\mathcal{EU}\cap\mathcal{AL}_{2}$) }is
NP-complete. This conjecture is confirmed in the rest of the present
paper.

\section{Main Result}

\subsection{Proof Outline}

We prove the NP-completeness of \noun{Syn($\mathcal{EU}\cap\mathcal{AL}_{2}$)}
by a polynomial reduction from \noun{3-SAT}. So, for arbitrary propositional
formula $\phi$ in 3-CNF we construct an Eulerian binary automaton
$A$ and a number $d$ such that
\begin{equation}
\phi\mbox{ is satisfiable \ensuremath{\Leftrightarrow A\mbox{ has a reset word of length \ensuremath{d}}}}.\label{eq:main}
\end{equation}
For the rest of the paper we fix a formula 
\[
\phi=\bigwedge_{i=1}^{m}\bigvee_{\lambda\in C_{i}}\lambda
\]
 on $n$ variables where each $C_{i}$ is a three-element set of literals,
i.e. subset of 
\[
L_{\phi}=\left\{ x_{1},\dots,x_{n},\neg x_{1},\dots,\neg x_{n}\right\} .
\]
We index the literals $\lambda\in L_{\Phi}$ by the following mapping
$\kappa$:\medskip{}

\begin{center}
\begin{tabular}{|>{\centering}p{0.8cm}||>{\centering}p{0.84cm}|>{\centering}p{0.84cm}|>{\centering}p{0.2cm}|>{\centering}p{0.84cm}|>{\centering}p{0.84cm}|>{\centering}p{0.84cm}|>{\raggedright}p{0.2cm}|>{\centering}p{0.84cm}|}
\hline 
$\lambda$ & $x_{1}$ & $x_{2}$ & $\!\!\dots$ & $x_{n}$ & $\neg x_{1}$ & $\neg x_{2}$ & $\!\!\dots$ & $\neg x_{n}$\tabularnewline
\hline 
$\kappa\!\left(\lambda\right)$ & $0$ & $1$ & $\!\!\dots$ & $n-1$ & $n$ & $n+1$ & $\!\!\dots$ & $2n-1$\tabularnewline
\hline 
\end{tabular}\medskip{}

\par\end{center}

Let $A=\left(Q,X,\delta\right)$, $X=\left\{ a,b\right\} $. Because
the structure of the automaton $A$ will be very heterogeneous, we
use an unusual method of description. The basic principles of the
method are:
\begin{itemize}
\item We describe the automaton $A$ via a labeled directed multigraph $G$,
representing the automaton in a standard way: edges of $G$ are labeled
by single letters $a$ and $b$ and carry the structure of the function
$\delta$. Paths in $G$ are thus \emph{labeled} by words from $\left\{ a,b\right\} ^{\star}$.
\item There is a collection of labeled directed multigraphs called \emph{templates}.
The graph $G$ is one of them. Another template is \texttt{SINGLE},
which consists of one vertex and no edges.
\item Each template \texttt{T}\textsf{$\neq$}\texttt{SINGLE} is expressed
in a fixed way as a disjoint union through a set \textsf{\textbf{$\PARTS_{\mathsf{\mathtt{T}}}$}}
of its proper subgraphs (the \emph{parts }of \texttt{T}), extended
by a set of additional edges (the \emph{links }of \texttt{T}). Each
$H\in\PARTS_{\mathsf{\mathtt{T}}}$ is isomorphic to some template
\texttt{U}. We say that $H$ \emph{is} \emph{of type }\texttt{U}\emph{.}
\item Let $q$ be a vertex of a template \texttt{T}, lying in a subgraph
$H\in\PARTS_{\mathtt{T}}$ which is of type \texttt{U} via a vertex
mapping $\rho:H\rightarrow\mathsf{\mathtt{U}}$. The \emph{local address
$\adr_{\mathsf{\mathtt{T}}}\!\left(q\right)$} is a finite string
of identifiers separated by ,,|''. It is defined inductively by
\[
\adr_{\mathsf{\mathtt{T}}}\!\left(q\right)=\begin{cases}
H\mid\adr_{\mathsf{\mathtt{U}}}\!\left(\rho\!\left(q\right)\right) & \mbox{if \ensuremath{\mathsf{\mathtt{U}}\neq\mathsf{\mathtt{SINGLE}}}}\\
H & \mbox{if \ensuremath{\mathsf{\mathtt{U}=\mathtt{SINGLE}}}}.
\end{cases}
\]
The string $\adr_{G}\!\left(q\right)$ is used as a regular vertex
identifier.
\end{itemize}
Having a word $w\in X^{\star}$, we denote a $t$-th letter of $w$
by $w_{t}$ and define the set $S_{t}=\delta\!\left(Q,w_{1}\dots w_{t}\right)$
of \emph{active states} \emph{at time $t$}. Whenever we depict a
graph, a solid arrow stands for the label $a$ and a dotted arrow
stands for the label $b$.

\subsection{Description of the Graph $G$}

\begin{flushleft}
Let us define all the templates and informally comment on their purpose.
Figure \ref{fig:abs} defines the template \texttt{ABS}, which does
not depend on the formula $\phi$.
\begin{figure}[H]
\begin{minipage}[t]{0.45\columnwidth}%
\begin{center}
\includegraphics{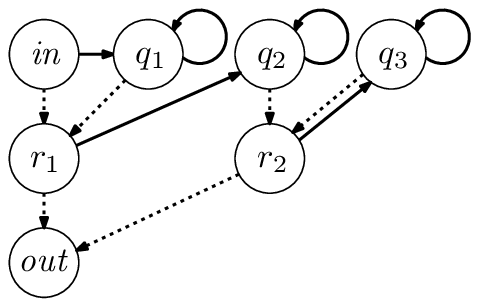}\caption{\label{fig:abs}Template \texttt{ABS}}

\par\end{center}%
\end{minipage}\hfill{}%
\begin{minipage}[t]{0.45\columnwidth}%
\begin{center}
\includegraphics{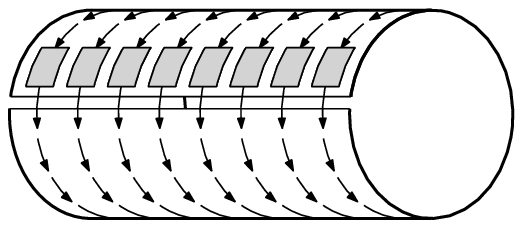}\caption{\label{fig:Cylinder}A barrier of \texttt{ABS} parts}

\par\end{center}%
\end{minipage}
\end{figure}

\par\end{flushleft}

The state $\mathit{out}$ of a part of type \texttt{ABS} is always
inactive after application of a word of length at least 2 which does
not contain $b^{2}$ as a factor. This allows us to ensure the existence
of a relatively short reset word. Actually, large areas of the graph
(namely the \texttt{CLAUSE($\dots$)} parts) have roughly the shape
depicted in Figure \ref{fig:Cylinder}, a cylindrical structure with
a horizontal barrier of \texttt{ABS }parts. If we use a sufficiently
long word with no occurrence of $b^{2}$, the edges outgoing from
the \texttt{ABS} parts are never used and almost all states become
inactive.

\begin{figure}[H]
\begin{centering}
\includegraphics{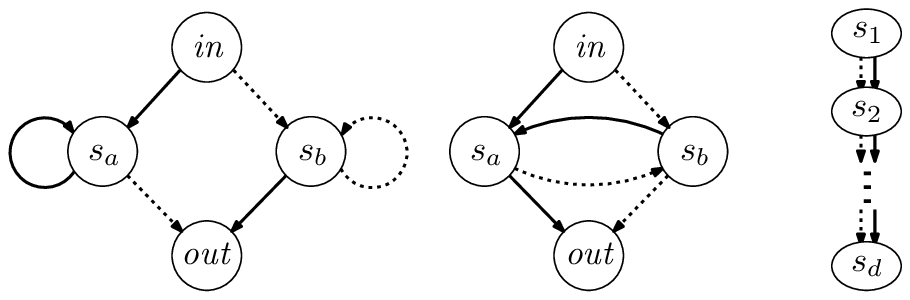}
\par\end{centering}

\caption{\label{fig:cca_cci_abs}Templates \texttt{CCA}, \texttt{CCI} and \texttt{PIPE($d$)}
respectively}
\end{figure}
Figure \ref{fig:cca_cci_abs} defines simple templates \texttt{CCA},
\texttt{CCI} and \texttt{PIPE($d$)} for each $d\geq1$. The activity
of an $\mathit{out}$ state depends on the last two letters applied.
In the case of \texttt{CCA} it is inactive if (and typically only
if) the two letters were equal. In the case of \texttt{CCI} it works
oppositely, equal letters correspond to active $\mathit{out}$ state.
One of the key ideas of the entire construction is the following.\emph{
}Let there be a subgraph of the form

\begin{equation}
\begin{array}{c}
\mbox{part }\mbox{of type \texttt{PIPE(\ensuremath{d})}}\\
\downarrow a,b\\
\mbox{part }\mbox{of type \texttt{CCA} or \texttt{CCI}}\\
\downarrow a,b\\
\mbox{part }\mbox{of type \texttt{PIPE(\ensuremath{d})}.}
\end{array}\label{eq:coders+pipes}
\end{equation}
Before the synchronization process starts, all the states are active.
As soon as the second letter of an input word is applied, the activity
of the $\mathit{out}$ state starts to depend on the last two letters
and the pipe below keeps a record of its previous activity. We say
that a part $H$ of type \texttt{PIPE($d$)} \emph{records }a sequence
$B_{1}\dots B_{d}\in\left\{ \mathbf{0},\mathbf{1}\right\} ^{d}$ \emph{at
time }$t$, if it holds that
\[
B_{k}=\mathbf{1}\Leftrightarrow H|s_{k}\notin S_{t}.
\]

In order to continue with defining templates, let us define a set
$M_{\phi}$ containing all the literals from $L_{\phi}$ and some
auxiliary symbols:
\[
M_{\phi}=L_{\phi}\cup\left\{ y_{1},\dots,y_{n}\right\} \cup\left\{ z_{1},\dots,z_{n}\right\} \cup\left\{ q,q',r,r'\right\} .
\]
We index the $4n+4$ members $\nu\in M_{\phi}$ by the following mapping
$\mu$:

\begin{flushleft}
\begin{tabular}{|>{\centering}p{0.8cm}||>{\centering}p{0.84cm}|>{\centering}p{0.84cm}|>{\centering}p{0.84cm}|>{\centering}p{0.84cm}|>{\centering}p{0.84cm}|>{\centering}p{0.84cm}|>{\raggedright}p{1.5mm}|>{\centering}p{0.84cm}|>{\centering}p{0.84cm}|}
\hline 
$\nu$ & $q$ & $r$ & $y_{1}$ & $x_{1}$ & $y_{2}$ & $x_{2}$ & $\!\!\dots$ & $y_{n}$ & $x_{n}$\tabularnewline
\hline 
$\mu\!\left(\nu\right)$ & $1$ & $2$ & $3$ & $4$ & $5$ & $6$ & $\!\!$ & $2n+1$ & $2n+2$\tabularnewline
\hline 
\end{tabular}
\par\end{flushleft}

\begin{flushleft}
\begin{tabular}{|>{\centering}p{0.8cm}||>{\centering}p{0.84cm}|>{\centering}p{0.84cm}|>{\centering}p{0.84cm}|>{\centering}p{0.84cm}|>{\centering}p{0.84cm}|>{\centering}p{0.84cm}|>{\raggedright}p{1.5mm}|>{\centering}p{0.84cm}|>{\centering}p{0.84cm}|}
\hline 
$\nu$ & $q'$ & $r'$ & $z_{1}$ & $\neg x_{1}$ & $z_{2}$ & $\neg x_{2}$ & $\!\!\dots$ & $z_{n}$ & $\neg x_{n}$\tabularnewline
\hline 
$\mu\!\left(\nu\right)$ & $2n+3$ & $2n+4$ & $2n+5$ & $2n+6$ & $2n+7$ & $2n+8$ & $\!\!\dots$ & $4n+3$ & $4n+4$\tabularnewline
\hline 
\end{tabular}
\par\end{flushleft}

The inverse mapping is denoted by $\mu'$. For each $\lambda\in L_{\phi}$
we define templates \texttt{INC($\lambda$)} and \texttt{NOTINC($\lambda$)},
both consisting of $12n+12$ \texttt{SINGLE} parts identified by elements
of $\left\{ 1,2,3\right\} \times M_{\phi}$. As depicted by Figure
\ref{fig:inc}, the links of \texttt{INC($\lambda$)}are:
\[
\begin{aligned}\left(1,\nu\right)\overset{a}{\longrightarrow} & \begin{cases}
\left(2,\lambda\right) & \mbox{if }\nu=\lambda\mbox{ or }\nu=r\\
\left(2,\nu\right) & \mbox{otherwise}
\end{cases}\\
\left(2,\nu\right)\overset{a}{\longrightarrow} & \begin{cases}
\left(3,q\right) & \mbox{if }\nu=r\mbox{ or }\nu=q\\
\left(3,\nu\right) & \mbox{otherwise}
\end{cases}
\end{aligned}
\qquad\begin{aligned}\left(1,\nu\right)\overset{b}{\longrightarrow} & \begin{cases}
\left(2,r\right) & \mbox{if }\nu=\lambda\mbox{ or }\nu=r\\
\left(2,\nu\right) & \mbox{otherwise}
\end{cases}\\
\left(2,\nu\right)\overset{b}{\longrightarrow} & \begin{cases}
\left(3,r\right) & \mbox{if }\nu=r\mbox{ or }\nu=q\\
\left(3,\nu\right) & \mbox{otherwise}
\end{cases}
\end{aligned}
\]
\begin{figure}
\begin{centering}
\subfloat[\label{fig:inc}\texttt{INC($\lambda$)}]{\begin{centering}
\includegraphics{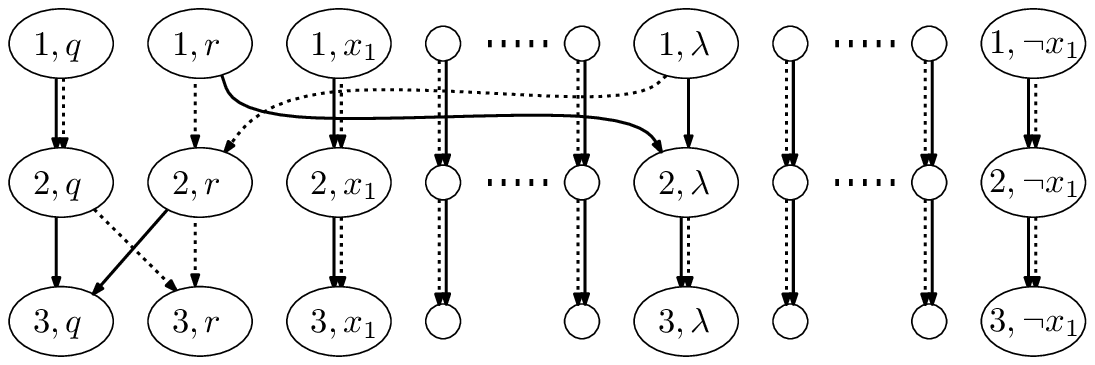}
\par\end{centering}

}\smallskip{}
\subfloat[\label{fig:notinc}\texttt{NOTINC($\lambda$)}]{\begin{centering}
\includegraphics{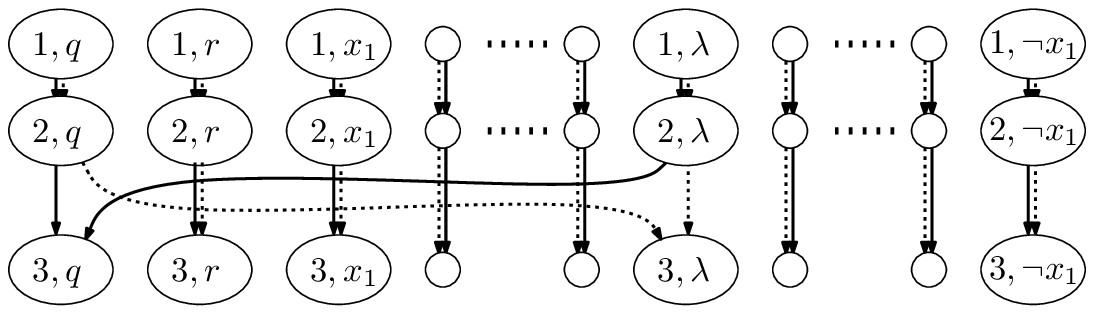}
\par\end{centering}

}
\par\end{centering}

\caption{Templates \texttt{INC($\lambda$)} and \texttt{NOTINC($\lambda$)}}
\end{figure}
\begin{figure}
\centering{}\includegraphics{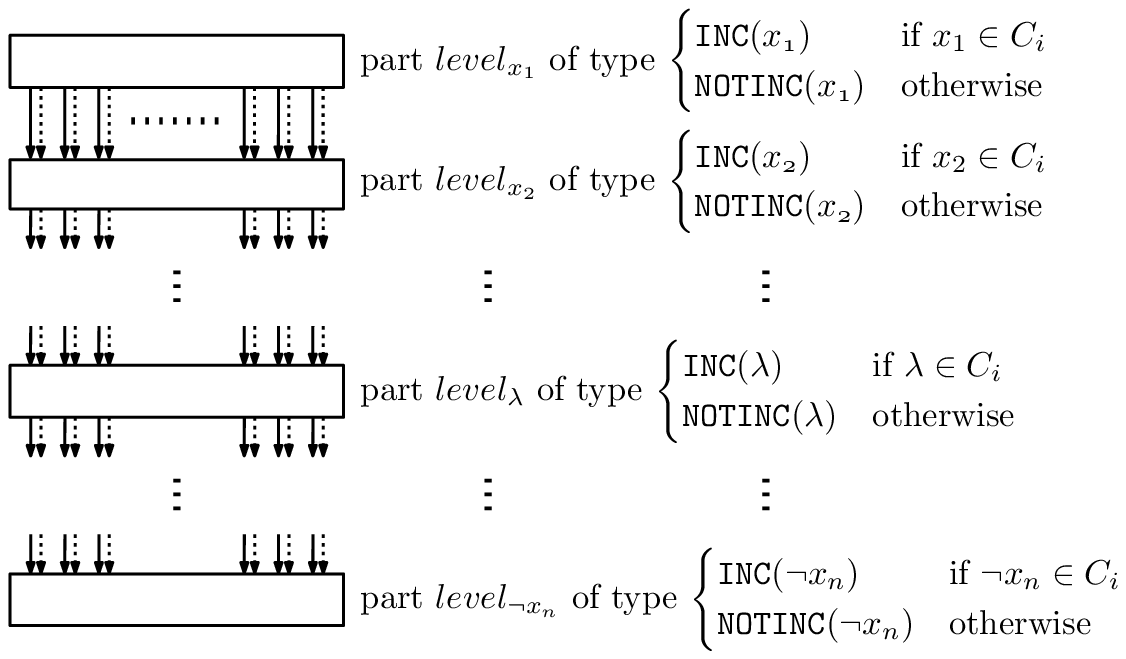}\caption{\label{fig:tester}Template \texttt{TESTER}}
\end{figure}
Note that we use the same identifier for an one-vertex subgraph and
for its vertex. As it is clear from Figure \ref{fig:notinc}, the
links of \texttt{NOTINC($\lambda$)} are: 
\[
\begin{aligned}\left(1,\nu\right)\overset{a}{\longrightarrow} & \,\left(2,\lambda\right)\\
\left(2,\nu\right)\overset{a}{\longrightarrow} & \begin{cases}
\left(3,q\right) & \mbox{if }\nu=q\mbox{ or }\nu=\lambda\\
\left(3,\nu\right) & \mbox{otherwise}
\end{cases}
\end{aligned}
\qquad\begin{aligned}\left(1,\nu\right)\overset{b}{\longrightarrow} & \,\left(2,r\right)\\
\left(2,\nu\right)\overset{b}{\longrightarrow} & \begin{cases}
\left(3,\lambda\right) & \mbox{if }\nu=q\mbox{ or }\nu=\lambda\\
\left(3,\nu\right) & \mbox{otherwise}
\end{cases}
\end{aligned}
\]

The key property of such templates comes to light when we need to
apply some two-letter word in order to make the state $\left(3,\lambda\right)$
inactive assuming $\left(1,r\right)$ inactive. If also $\left(1,\lambda\right)$
is initially inactive, we can use the word $a^{2}$ in both templates.
If it is active (which corresponds to the idea of unsatisfied literal
$\lambda$), we discover the difference between the two templates:
The word $a^{2}$ works if the type is \texttt{NOTINC($\lambda$)},
but fails in the case of \texttt{INC($\lambda$)}. Such\emph{ failure
}corresponds to the idea of unsatisfied literal $\lambda$ occurring
in a clause of $\phi$.

For each clause (each $i\in\left\{ 1,\dots,m\right\} $) we define
a template \texttt{TESTER($i$)}. It consists of $2n$ serially linked
parts, namely $\mathit{level}_{\lambda}$ for each $\lambda\in L_{\phi}$,
each of type \texttt{INC($\lambda$)} or \texttt{NOTINC($\lambda$)}.
The particular type of each $\mathit{level}_{\lambda}$ depends on
the clause $C_{i}$ as seen in Figure \ref{fig:tester}, so exactly
three of them are always of type \texttt{INC($\dots$)}. If the corresponding
clause is unsatisfied, each of its three literals is unsatisfied,
which causes three\emph{ }failures within the levels. Three failures
imply at least three occurrences of $b$, which turns up to be too
much for a reset word of certain length to exist. Clearly we still
need some additional mechanisms to realize this vague vision.

Figure \ref{fig:forcer_limiter} defines templates \texttt{FORCER}
and \texttt{LIMITER}. The idea of template \texttt{FORCER} is simple.
Imagine a situation when $q_{1,0}$ or $r_{1,0}$ is active and we
need to deactivate the entire forcer by a word of length at most $2n+3$.
Any use of $b$ would cause an unbearable delay, so if such a word
exists, it starts by $a^{2n+2}$.

\begin{flushleft}
The idea of \texttt{LIMITER} is similar, but we tolerate some occurrences
of $b$ here, namely two of them. This works if we assume $s_{1,0}$
active and it is necessary to deactivate the entire limiter by a word
of length at most $6n+1$. 
\begin{figure}
\begin{centering}
\includegraphics{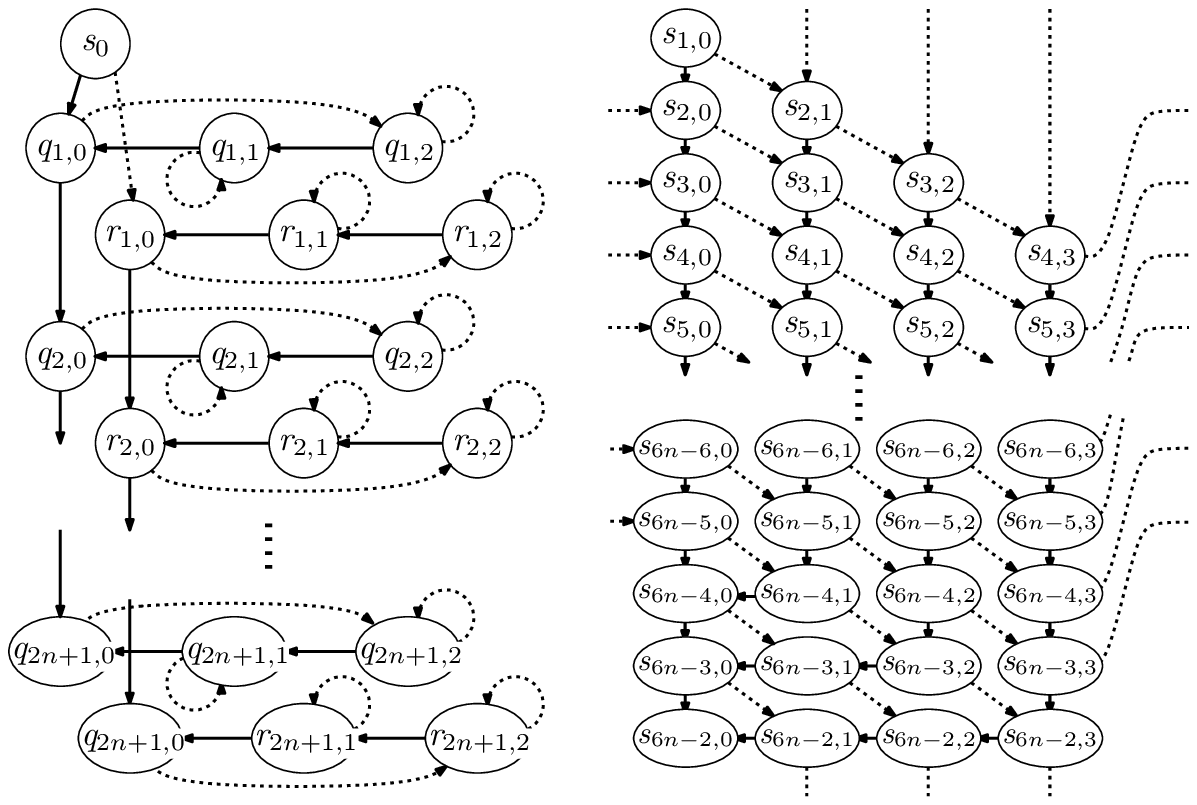}
\par\end{centering}

\caption{\label{fig:forcer_limiter}Templates \texttt{FORCER} and \texttt{LIMITER}
respectively}
\end{figure}
 
\par\end{flushleft}

We also need a template \texttt{PIPES($d,k$)} for each $d,k\geq1$.
It consists just of $k$ parallel pipes of length $d$. Namely there
is a \texttt{SINGLE} part $s_{d',k'}$ for each $d'\leq d$, $k'\leq k$
and all the edges are of the form $s_{d',k'}\longrightarrow s_{d'+1,k'}$.

\begin{flushleft}
The most complex templates are \texttt{CLAUSE($i$)} for each $i\in\left\{ 1,\dots,m\right\} $.
Denote
\begin{eqnarray*}
\alpha_{i} & = & \left(i-1\right)\left(12n-2\right),\\
\beta_{i} & = & \left(m-i\right)\left(12n-2\right).
\end{eqnarray*}
\begin{figure}
\begin{centering}
\includegraphics{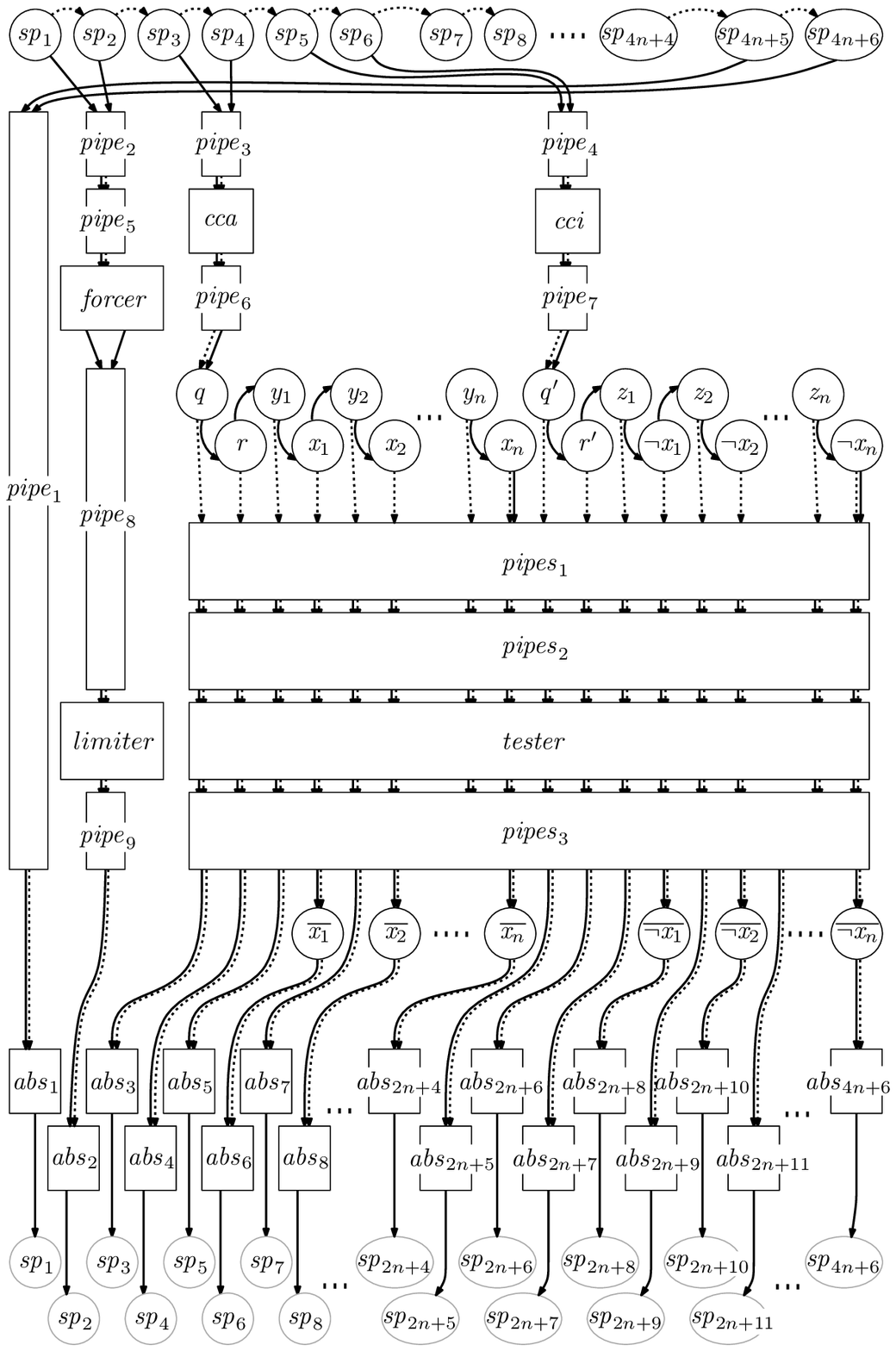}
\par\end{centering}

\caption{\label{fig:clause}Template \texttt{CLAUSE($i$)}}
\end{figure}
As shown in Figure \ref{fig:clause}, \texttt{CLAUSE($i$)} consists
of the following parts:
\par\end{flushleft}
\begin{itemize}
\item Parts $\mathit{sp}_{1},\dots,\mathit{sp}_{4n+6}$ of type \texttt{SINGLE}.
\item Parts $\mathit{abs}_{1},\dots,\mathit{abs}_{4n+6}$ of type \texttt{ABS}.
The entire template has a shape similar to Figure \ref{fig:Cylinder},
including the barrier of \texttt{ABS} parts.
\item Parts $\mathit{pipe}_{2}$, $\mathit{pipe}_{3}$, $\mathit{pipe}_{4}$
of types \texttt{PIPE($2n-1$)} and $\mathit{pipe}_{6}$, $\mathit{pipe}_{7}$
of types \texttt{PIPE($2n+2$)}.
\item Parts $\mathit{cca}$ and $\mathit{cci}$ of types \texttt{CCA} and
\texttt{CCI} respectively. Together with the pipes above they realize
the idea described in (\ref{eq:coders+pipes}). As they form two constellations
which work simultaneously, the parts $\mathit{pipe}_{6}$ and $\mathit{pipe}_{7}$
typically record mutually inverse sequences. We interpret them as
an assignment of the variables $x_{1},\dots,x_{n}$. Such assignment
is then processed by the tester.
\item A part $\nu$ of type \texttt{SINGLE} for each $\nu\in M_{\phi}$. 
\item A part $\mathit{tester}$ of type \texttt{TESTER($i$)}.
\item A part $\overline{\lambda}$ of type \texttt{SINGLE} for each $\lambda\in L_{\phi}$.
While describing the templates \texttt{INC($\lambda$)} and \texttt{NOTINC($\lambda$)}
we claimed that in certain case there arises a need to make the state
$\left(3,\lambda\right)$ inactive. This happens when the border of
inactive area moves down through the tester levels. The point is that
any word of length $6n$ deactivates the entire tester, but we need
to ensure that some tester columns, namely the $\kappa\!\left(\lambda\right)$-th
for each $\lambda\in L_{\phi}$, are deactivated one step earlier.
If some of them is still active just before the deactivation of tester
finishes, the state $\overline{\lambda}$ becomes active, which slows
down the synchronization process.
\item Parts $\mathit{pipes}_{1}$, $\mathit{pipes}_{2}$ and $\mathit{pipes}_{3}$
of types \texttt{PIPES($\alpha_{i},4n+4$)}, \texttt{PIPES($6n-2,4n+4$)}
and \texttt{PIPES($\beta_{i},4n+4$)} respectively. There are multiple
clauses in $\phi$, but multiple testers cannot work in parallel.
That is why each of them is padded by a passive \texttt{PIPES($\dots$)}
part of size depending on particular $i$. If $\alpha_{i}=0$ or $\beta_{i}=0$,
the corresponding \texttt{PIPES} part is not present in $\mathit{cl}_{i}$.
\item Parts $\mathit{pipe}_{1}$, $\mathit{pipe}_{5}$, $\mathit{pipe}_{8}$,
$\mathit{pipe}_{9}$ of types \texttt{PIPE($12mn+4n-2m+6$)}, \texttt{PIPE($4$)},
\texttt{PIPE($\alpha_{i}+6n-1$)}, \texttt{PIPE($\beta_{i}$)} respectively.
\item The part $\mathit{forcer}$ of type \texttt{FORCER}. This part guarantees
that only the letter $a$ is used in certain segment of the word $w$.
This is necessary for the data produced by $\mathit{cca}$ and $\mathit{cci}$
to safely leave the parts $\mathit{pipe}_{3}$, $\mathit{pipe}_{4}$
and line up in the states of the form $\nu$ for $\nu\in M_{\phi}$,
from where they are shifted to the tester.
\item The part $limiter$ of type \texttt{LIMITER}. This part guarantees
that the letter $b$ occurs at most twice when the border of inactive
area passes through the tester. Because each unsatisfied literal from
the clause requests an occurrence of $b$, only a satisfied clause
meets all the conditions for a reset word of certain length to exist.
\end{itemize}
Links of $\mathtt{CLAUSE(}i\mathtt{)}$, which are not clear from
Figure \ref{fig:clause} are
\[
\begin{aligned}\nu\overset{a}{\longrightarrow} & \,\begin{cases}
\mathit{pipes}_{1}|s_{1,\mu\left(\nu\right)} & \mbox{if }\nu=\neg x_{n}\\
\mu'\!\left(\mu\!\left(\nu\right)+1\right) & \mbox{otherwise}
\end{cases}\end{aligned}
\qquad\begin{aligned}\nu\overset{b}{\longrightarrow} & \,\mathit{pipes}_{1}|s_{1,\mu\left(\nu\right)}\end{aligned}
\]
for each $\nu\in M_{\phi}$ and
\[
\begin{aligned}\mathit{pipes}_{3}|s_{\beta_{i},k}\overset{a,b}{\longrightarrow} & \begin{cases}
\overline{\mu'\!\left(k\right)} & \mbox{if }\mu'\!\left(k\right)\in L_{\phi}\\
\mathit{abs}_{k+2}|\mathit{in} & \mbox{otherwise}
\end{cases}\end{aligned}
\qquad\begin{aligned}\overline{\lambda}\overset{a,b}{\longrightarrow} & \, abs_{\mu\left(\lambda\right)+2}|\mathit{in}\end{aligned}
\]
for each $k\in\left\{ 1,\dots,4n+4\right\} $, $\lambda\in L_{\phi}$.

We are ready to form the whole graph $G$, see Figure \ref{fig:wholeG}.
For each $i,k\in\left\{ 1,\dots m\right\} $ there are parts $\mathit{cl}_{k},\mathit{abs}_{k}$
of types \texttt{CLAUSE($i$)} and \texttt{ABS} respectively and parts
$q_{k},r_{k},r'_{k},s_{1},s_{2}$ of type \texttt{SINGLE}. The edge
incoming to a $\mathit{cl}_{i}$ part ends in $\mathit{cl}_{i}|\mathit{sp}_{1}$,
the outgoing one starts in $\mathit{cl}_{i}|\mathit{sp}_{4n+6}$.
When no states outside \texttt{ABS} parts are active within each \texttt{CLAUSE($\dots$)}
part and no $\mathit{out}$, $r_{1}$ nor $r_{2}$ state is active
in any \texttt{ABS} part, the word $b^{2}ab^{4n+m+7}$ takes all active
states to $s_{2}$ and completes the synchronization.
\begin{figure}
\begin{centering}
\includegraphics{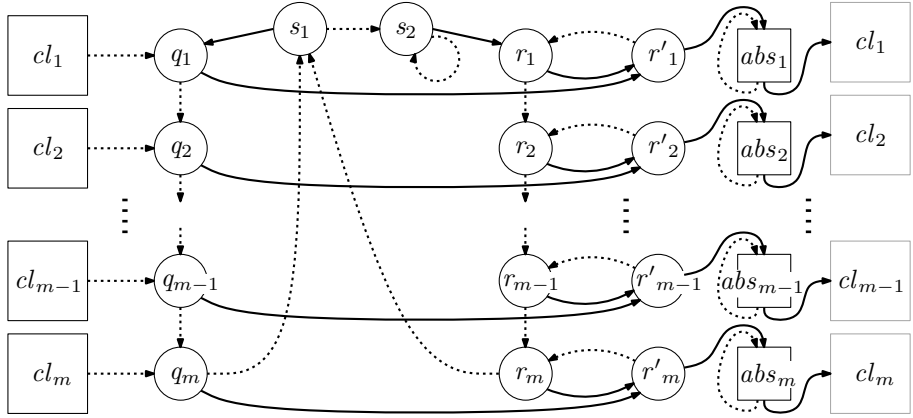}
\par\end{centering}

\caption{\label{fig:wholeG}The graph $G$}
\end{figure}
 Graph $G$ does not fully represent the automaton $A$ yet because
there are
\begin{itemize}
\item $8mn+4m$ vertices with only one outgoing edge, namely $\mathit{cl}_{i}|\mathit{abs}_{k}|\mathit{out}$
and $\mathit{cl}_{i}|\mathit{sp}_{l}$ for each $i\in\left\{ 1,\dots,m\right\} ,k\in\left\{ 1,\dots,4n+6\right\} ,l\in\left\{ 7,\dots,4n+4\right\} $,
\item $8mn+4m$ vertices with only one incoming edge: $\mathit{cl}_{i}|\nu$
and $\mathit{cl}_{i}|\mathit{pipes}_{1}|\left(1,\nu'\right)$ for
each $i\in\left\{ 1,\dots,m\right\} ,\nu\in M_{\phi}\backslash\left\{ q,q'\right\} ,\nu'\in M_{\phi}\backslash\left\{ x_{n},\neg x_{n}\right\} $. 
\end{itemize}
But we do not need to specify the missing edges exactly, let us just
say that they somehow connect the relevant states and the automaton
$A$ is complete. Let us set 
\[
d=12mn+8n-m+18
\]
 and prove that the equivalence (\ref{eq:main}) holds.

\subsection{From an Assignment to a Word}

First let us suppose that there is an assignment $\xi_{1},\dots,\xi_{n}\in\left\{ \mathbf{0},\mathbf{1}\right\} $
of the variables $x_{1},\dots,x_{n}$ (respectively) satisfying the
formula $\phi$ and prove that the automaton $A$ has a reset word
$w$ of length $d$. For each $j\in\left\{ 1,\dots,n\right\} $ we
denote 
\[
\sigma_{j}=\begin{cases}
a & \mbox{if }\xi_{j}=\mathbf{1}\\
b & \mbox{if }\xi_{j}=\mathbf{0}
\end{cases}
\]
 and for each $i\in\left\{ 1,\dots,m\right\} $ we choose a satisfied
literal $\overline{\lambda}_{i}$ from $C_{i}$. We set
\[
w=a^{2}\left(\sigma_{n}a\right)\left(\sigma_{n-1}a\right)\dots\left(\sigma_{1}a\right)aba^{2n+3}b\left(a^{6n-2}v_{1}\right)\dots\left(a^{6n-2}v_{m}\right)b^{2}ab^{4n+m+7},
\]
where for each $i\in\left\{ 1,\dots,m\right\} $ we use the word 
\[
v_{i}=u_{i,x_{1}}\dots u_{i,x_{n}}u_{i,\neg x_{1}}\dots u_{i,\neg x_{n}},
\]
denoting
\[
u_{i,\lambda}=\begin{cases}
a^{3} & \mbox{if }\lambda=\overline{\lambda}_{i}\mbox{ or }\lambda\notin C_{i}\\
ba^{2} & \mbox{if }\lambda\neq\overline{\lambda}_{i}\mbox{ and }\lambda\in C_{i}
\end{cases}
\]
for each $\lambda\in L_{\phi}$. We see that $\left|v_{i}\right|=6n$
and therefore 
\[
\left|w\right|=4n+8+m\left(12n-2\right)+4n+m+10=12mn+8n-m+18=d.
\]
Let us denote 
\[
\gamma=12mn+4n-2m+9
\]
and 
\[
\overline{S}_{t}=Q\backslash S_{t}
\]
 for each $t\leq d$. Because the first occurrence of $b^{2}$ in
$w$ starts by the $\gamma$-th letter, we have:
\begin{lem}
\label{lem: no abs out}Each state of a form $\mathit{cl}_{\dots}|\mathit{abs}_{\dots}|\mathit{out}$
or $\mathit{abs_{\dots}}|\mathit{out}$ lies in $\overline{S}_{2}\cap\dots\cap\overline{S}_{\gamma}$.
\end{lem}
Let us fix an arbitrary $i\in\left\{ 1,\dots,m\right\} $ and describe
a growing area of inactive states within $\mathit{cl}_{i}$. We use
the following method of verifying inactivity of states: Having a state
$s\in Q$ and $t,k\geq1$ such that any path of length $k$ ending
in $s$ uses a member of $\overline{S}_{t-k}\cap\dots\cap\overline{S}_{t-1}$,
we easily deduce that $s\in\overline{S}_{t}$. In such case let us
just say that $k$ \emph{witnesses} that $s\in\overline{S}_{t}$.
The following claims follow directly from the definition of $w$.
Note that Claim 7 relies on the fact that $b$ occurs only twice in
$v_{i}$.
\begin{lem}
\ \label{lem: growing inactivity}\emph{}\\
\emph{}%
\begin{tabular}{llll}
\noalign{\vskip0.2cm}
\emph{1.} & \emph{$\left\{ \mathit{cl}_{i}|\mathit{sp}_{1},\dots,\mathit{cl}_{i}|\mathit{sp}_{4n+6}\right\} $ } & $\subseteq$ & \emph{$\overline{S}_{2}\cap\dots\cap\overline{S}_{\gamma}$}\tabularnewline
\noalign{\vskip0.2cm}
\emph{2.} & \emph{$\mathit{cl}_{i}|\mathit{pipe}_{2}\cup\mathit{cl}_{i}|\mathit{pipe}_{3}\cup\mathit{cl}_{i}|\mathit{pipe}_{4}$} & $\subseteq$ & \emph{$\overline{S}_{2n+1}\cap\dots\cap\overline{S}_{\gamma}$}\tabularnewline
\noalign{\vskip0.2cm}
\emph{3.} & \emph{$\mathit{cl}_{i}|\mathit{cca}\cup\mathit{cl}_{i}|\mathit{cci}\cup\mathit{cl}_{i}|pipe_{5}$} & $\subseteq$ & $\overline{S}_{2n+5}\cap\dots\cap\overline{S}_{\gamma}$\tabularnewline
\noalign{\vskip0.2cm}
\emph{4.} & \emph{$\mathit{cl}_{i}|\mathit{pipe}_{6}\cup\mathit{cl}_{i}|\mathit{pipe}_{7}\cup\mathit{cl}_{i}|forcer$} & $\subseteq$ & \emph{$\overline{S}_{4n+7}\cap\dots\cap\overline{S}_{\gamma}$}\tabularnewline
\noalign{\vskip0.2cm}
\emph{5.} & \emph{$\left\{ \mathit{cl}_{i}|\nu\,:\,\nu\in M_{\phi}\right\} $} & $\subseteq$ & \emph{$\overline{S}_{4n+8}\cap\dots\cap\overline{S}_{\gamma}$}\tabularnewline
\noalign{\vskip0.2cm}
\emph{6.} & $\mathit{cl}_{i}|\mathit{pipes}_{1}\cup\mathit{cl}_{i}|\mathit{pipes}_{2}\cup\mathit{cl}_{i}|\mathit{pipe}_{8}$ & $\subseteq$ & $\overline{S}_{10n+\alpha_{i}+6}\cap\dots\cap\overline{S}_{\gamma}$\tabularnewline
\noalign{\vskip0.2cm}
\emph{7.} & \emph{$\mathit{cl}_{i}|\mathit{limiter}\cup\mathit{cl}_{i}|\mathit{tester}$} & $\subseteq$ & \emph{$\overline{S}_{16n+\alpha_{i}+6}\cap\dots\cap\overline{S}_{\gamma}$}\tabularnewline
\noalign{\vskip0.2cm}
\emph{8.} & $\mathit{cl}_{i}|\mathit{pipe}_{1}\cup\mathit{cl}_{i}|\mathit{pipe}_{9}\cup\mathit{cl}_{i}|\mathit{pipes}_{3}$  & $\subseteq$ & \emph{$\overline{S}_{\gamma-1}\cap\overline{S}_{\gamma}$}\tabularnewline
\end{tabular}\end{lem}
\begin{proof}
~
\begin{enumerate}
\item \emph{Claim: $\left\{ \mathit{cl}_{i}|\mathit{sp}_{1},\dots,\mathit{cl}_{i}|\mathit{sp}_{4n+6}\right\} \subseteq\overline{S}_{2}\cap\dots\cap\overline{S}_{\gamma}$
}.\\
We have $w_{1}w_{2}=a^{2}$ and there is no path labeled by $a^{2}$
ending in any $\mathit{cl}_{i}|\mathit{sp}_{\dots}$ state, so such
states lie in $\overline{S}_{2}$. For each $t=3,\dots,\gamma$ we
can inductively use $k=1$ to witness the memberships in $\overline{S}_{t}$.
In the induction step we use Lemma \ref{lem: no abs out}, which excludes
the $\mathit{out}$ states of the \texttt{ABS} parts from each corresponding
$\overline{S}_{t-1}$.\smallskip{}

\item \emph{Claim: $\mathit{cl}_{i}|\mathit{pipe}_{2}\cup\mathit{cl}_{i}|\mathit{pipe}_{3}\cup\mathit{cl}_{i}|\mathit{pipe}_{4}\subseteq\overline{S}_{2n+1}\cap\dots\cap\overline{S}_{\gamma}$}.\\
All the memberships are witnessed by $k=2n-1$, because any path of
the length $2n-1$ ending in such state must use a $\mathit{cl}_{i}|\mathit{sp}_{\dots}$
state and such states lie in $\overline{S}_{2}\cap\dots\cap\overline{S}_{\gamma}$
by the previous claim.\smallskip{}

\item \emph{Claim: $\mathit{cl}_{i}|\mathit{cca}\cup\mathit{cl}_{i}|\mathit{cci}\cup\mathit{cl}_{i}|pipe_{5}\subseteq\overline{S}_{2n+5}\cap\dots\cap\overline{S}_{\gamma}$}.\\
We have $w_{2n+2}\dots w_{2n+5}=a^{2}ba$, which clearly maps each
state of \emph{$\mathit{cl}_{i}|\mathit{cca}$},\emph{ $\mathit{cl}_{i}|\mathit{cci}$}
or \emph{$\mathit{cl}_{i}|pipe_{5}$ }out of those parts. Each path
of length $4$ leading into the parts from outside starts in $\overline{S}_{2n+1}$,
so it follows that all the states lie in $\overline{S}_{2n+5}$. To
prove the rest we inductively use the witness $k=1$.\smallskip{}

\item \emph{Claim: $\mathit{cl}_{i}|\mathit{pipe}_{6}\cup\mathit{cl}_{i}|\mathit{pipe}_{7}\cup\mathit{cl}_{i}|forcer\subseteq\overline{S}_{4n+7}\cap\dots\cap\overline{S}_{\gamma}$}.\\
In the cases of \emph{$\mathit{cl}_{i}|\mathit{pipe}_{6}$ and $\mathit{cl}_{i}|\mathit{pipe}_{7}$}
we just use the witness $k=2n+2$. In the case of \emph{$\mathit{cl}_{i}|forcer$
}we proceed the same way as in the previous claim. We have $w_{2n+6}\dots w_{4n+7}=a^{2n+2}$.
Because also $w_{2n+5}=a$, only the states $q_{\dots,0}$ can be
active within the part $\mathit{cl}_{i}|forcer$ in time $2n+6$.
The word $w_{2n+7}\dots w_{4n+7}$ maps all such states out of $\mathit{cl}_{i}|\mathit{forcer}$.
Each path of length $2n+2$ leading into $\mathit{cl}_{i}|\mathit{forcer}$
from outside starts in $\overline{S}_{2n+5}$, so it follows that
all states from $\mathit{cl}_{i}|\mathit{forcer}$ lie in $\overline{S}_{4n+7}$.
To handle $t=4n+8,\dots,\gamma$ we inductively use the witness $k=1$.\smallskip{}

\item \emph{Claim: $\left\{ \mathit{cl}_{i}|\nu\,:\,\nu\in M_{\phi}\right\} \subseteq\overline{S}_{4n+8}\cap\dots\cap\overline{S}_{\gamma}$}.\\
In the cases of $\mathit{cl}_{i}|q$ and $\mathit{cl}_{i}|q'$ we
use the witness $1$. We have $w_{4n+8}=b$ and the only edges labeled
by $b$ incoming to remaining states could be some of the $8mn+4m$
unspecified edges of $G$. But we have $w_{4n+6}w_{4n+7}=a^{2}$,
so each $\mathit{out}$ state of any \texttt{ABS} part lies in $\overline{S}_{4n+7}$
and thus no unspecified edge starts in a state outside $\overline{S}_{4n+7}$.\smallskip{}

\item \emph{Claim: }$\mathit{cl}_{i}|\mathit{pipes}_{1}\cup\mathit{cl}_{i}|\mathit{pipes}_{2}\cup\mathit{cl}_{i}|\mathit{pipe}_{8}\subseteq\overline{S}_{10n+\alpha_{i}+6}\cap\dots\cap\overline{S}_{\gamma}$.\\
We use witnesses $k=\alpha_{i}$ for \emph{$\mathit{cl}_{i}|\mathit{pipes}_{1}$},\emph{
$k=6n-2$} for \emph{$\mathit{cl}_{i}|\mathit{pipes}_{2}$} and $k=\alpha_{i}+6n-1$
for \emph{$\mathit{cl}_{i}|\mathit{pipe}_{8}$.}\smallskip{}

\item \emph{Claim: $\mathit{cl}_{i}|\mathit{limiter}\cup\mathit{cl}_{i}|\mathit{tester}\subseteq\overline{S}_{16n+\alpha_{i}+6}\cap\dots\cap\overline{S}_{\gamma}$}.\\
Because 
\[
w_{4n+\alpha_{i}+9}\dots w_{10n+\alpha_{i}+6}=a^{6n-2},
\]
there are only states of the form $\mathit{cl}_{i}|\mathit{limiter}|s_{\dots,0}$
in the intersection of $\mathit{cl}_{i}|\mathit{limiter}$ and $S_{10n+\alpha_{i}+6}$.
Together with the fact that there are only two occurrences of $b$
in $v_{i}$ it confirms that the case of $\mathit{cl}_{i}|\mathit{limiter}$
holds. The case of \emph{$\mathit{cl}_{i}|\mathit{tester}$} is easily
witnessed by $k=6n$. \smallskip{}

\item \emph{Claim: }$\mathit{cl}_{i}|\mathit{pipe}_{1}\cup\mathit{cl}_{i}|\mathit{pipe}_{9}\cup\mathit{cl}_{i}|\mathit{pipes}_{3}\subseteq\overline{S}_{\gamma-1}\cap\overline{S}_{\gamma}$
.\\
We use witnesses $k=12mn+4n-2m+6$ for \emph{$\mathit{cl}_{i}|\mathit{pipe}_{1}$
}and $k=\beta_{i}$ for \emph{$\mathit{cl}_{i}|\mathit{pipe}_{9},\mathit{cl}_{i}|\mathit{pipes}_{3}$.}
\end{enumerate}
\end{proof}
For each $\lambda\in L_{\phi}$ we ensure by the word $u_{i,\lambda}$
that the $\kappa\!\left(\lambda\right)$-th tester column is deactivated
in advance, namely at time $t=16n+\alpha_{i}+5$. The advance allows
the following key claim to hold true.
\begin{lem}
\label{lem: no lambda bar}\emph{$\left\{ \mathit{cl}_{i}|\overline{\lambda}\,:\,\lambda\in L_{\phi}\right\} \subseteq\overline{S}_{\gamma-1}\cap\overline{S}_{\gamma}$.}\end{lem}
\begin{proof}
For each such $\lambda$ we choose 
\[
k=6n-3\kappa\!\left(\lambda\right)+\beta_{i}+1
\]
 as a witness of $\mathit{cl}_{i}|\overline{\lambda}\in\overline{S}_{\gamma-1}$.
There is only one state where a path of length $k$ ending in $\overline{\lambda}$
starts: the state 
\[
s=\mathit{cl}_{i}|\mathit{tester}|\mathit{level}_{\lambda}|\left(3,\lambda\right).
\]
It holds that 
\[
s\in\overline{S}_{10n+\alpha_{i}+3\kappa\left(\lambda\right)+6}\cap\dots\cap\overline{S}_{\gamma},
\]
as is easily witnessed by $k'=3\kappa\!\left(\lambda\right)$ using
Claim 6 of Lemma \ref{lem: growing inactivity}. But we are going
to show also that 
\begin{equation}
s\in\overline{S}_{10n+\alpha_{i}+3\kappa\left(\lambda\right)+5},\label{eq: deactivate the column}
\end{equation}
which will imply that $k$ is a true witness of $\overline{\lambda}\in\overline{S}_{\gamma-1}$,
because
\[
\left(\gamma-1\right)-k=10n+\alpha_{i}+3\kappa\!\left(\lambda\right)+5.
\]
So let us prove the membership (\ref{eq: deactivate the column}).
We need to observe, using the definition of $w$, that:
\begin{itemize}
\item At time $2n+5$ the part $\mathit{pipe}_{6}$ records the sequence
\[
\mathbf{0},\mathbf{1},\xi_{1},\xi_{1},\xi_{2},\xi_{2},\dots,\xi_{n},\xi_{n}
\]
and the part $\mathit{pipe}_{7}$ records the sequence of inverted
values. Because 
\[
w_{2n+6}\dots w_{4n+7}=a^{2n+2},
\]
at time $4n+7$ the states $q,r'$ are active, the states $q',r$
are inactive and for each $j\in\left\{ 1,\dots,n\right\} $ it holds
that
\[
x_{j}\in\overline{S}_{4n+7}\Leftrightarrow y_{j}\in\overline{S}_{4n+7}\Leftrightarrow\neg x_{j}\in S_{4n+7}\Leftrightarrow z_{j}\in S_{4n+7}\Leftrightarrow\xi_{j}=\mathbf{1}.
\]
Because $w_{4n+8}=b$, at time $10n+\alpha_{i}+6$ we find the whole
structure above shifted to the first row of $\mathit{cl}_{i}|\mathit{tester}$,
so particularly for $\lambda\in L_{\phi}$:
\[
\mathit{cl}_{i}|\mathit{tester}|\mathit{level}_{x_{1}}|\left(1,\lambda\right)\in\overline{S}_{10n+\alpha_{i}+6}\Leftrightarrow\lambda\mbox{ is satisfied by }\xi_{1},\dots,\xi_{n}.
\]

\item From a simple induction on tester levels it follows that
\[
\mathit{cl}_{i}|\mathit{tester}|\mathit{level}_{\lambda}|\left(1,r\right)\in\overline{S}_{10n+\alpha_{i}+3\kappa\left(\lambda\right)+3}.
\]

\end{itemize}
Note that 
\[
w_{10n+\alpha_{i}+3\kappa\left(\lambda\right)+4}w_{10n+\alpha_{i}+3\kappa\left(\lambda\right)+5}w_{10n+\alpha_{i}+3\kappa\left(\lambda\right)+6}=u_{i,\lambda}
\]
and distinguish the following cases:
\begin{itemize}
\item If $\lambda=\overline{\lambda}_{i}$, we have $\lambda\in C_{i}$,
the part $\mathit{cl}_{i}|\mathit{tester}|\mathit{level}_{\lambda}$
is of type \texttt{INC($\lambda$)} and $u_{i,\lambda}=a^{3}$. We
also know that $\lambda$ is satisfied, so
\[
\mathit{cl}_{i}|\mathit{tester}|\mathit{level}_{x_{1}}|\left(1,\lambda\right)\in\overline{S}_{10n+\alpha_{i}+6}.
\]
The state above is the only state, from which any path of length $3\kappa\!\left(\lambda\right)-3$
leads to $\mathit{cl}_{i}|\mathit{tester}|\mathit{level}_{\lambda}|\left(1,\lambda\right)$,
so we deduce that
\[
\mathit{cl}_{i}|\mathit{tester}|\mathit{level}_{\lambda}|\left(1,\lambda\right)\in\overline{S}_{10n+\alpha_{i}+3\kappa\left(\lambda\right)+3}.
\]
We see that each path labeled by $a^{2}$ ending in $\mathit{cl}_{i}|\mathit{tester}|\mathit{level}_{\lambda}|\left(3,\lambda\right)$
starts in $\mathit{cl}_{i}|\mathit{tester}|\mathit{level}_{\lambda}|\left(1,\lambda\right)$
or in $\mathit{cl}_{i}|\mathit{tester}|\mathit{level}_{\lambda}|\left(1,r\right)$,
but each of the two states lies in $\overline{S}_{10n+\alpha_{i}+3\kappa\left(\lambda\right)+3}$.
So the membership (\ref{eq: deactivate the column}) holds.
\item If $\lambda\notin C$, the part $\mathit{cl}_{i}|\mathit{tester}|\mathit{level}_{\lambda}$
is of type \texttt{NOTINC($\lambda$)} and $u_{i,\lambda}=a^{3}$.
\linebreak{}
Particularly $w_{10n+\alpha_{i}+3\kappa\left(\lambda\right)+5}=a$
but no edge labeled by $a$ comes to \linebreak{}
$\mathit{cl}_{i}|\mathit{tester}|\mathit{level}_{\lambda}|\left(3,\lambda\right)$
and the membership (\ref{eq: deactivate the column}) follows trivially.
\item If $\lambda\neq\overline{\lambda}_{i}$ and $\lambda\in C_{i}$, the
part $\mathit{cl}_{i}|\mathit{tester}|\mathit{level}_{\lambda}$ is
of type \texttt{INC($\lambda$)} and $u_{i,\lambda}=ba^{2}$. Particularly
\[
w_{10n+\alpha_{i}+3\kappa\left(\lambda\right)+4}w_{10n+\alpha_{i}+3\kappa\left(\lambda\right)+5}=ba,
\]
but no path labeled by $ba$ comes to $\mathit{cl}_{i}|\mathit{tester}|\mathit{level}_{\lambda}|\left(3,\lambda\right)$,
so we reach the same conclusion as in the previous case.
\end{itemize}
We have proven that $\mathit{cl}_{i}|\overline{\lambda}$ lies in
$\overline{S}_{\gamma-1}$. From Claim 8 of Lemma \ref{lem: growing inactivity}
it follows directly that it lies also in $\overline{S}_{\gamma}$.
\end{proof}
We see that within $\mathit{cl}_{i}$ only states from the \texttt{ABS}
parts can lie in $S_{\gamma-1}$. Since $w_{\gamma-2}w_{\gamma-1}=a^{2}$,
no state $r_{1}$, $r_{2}$ or $\mathit{out}$ from any \texttt{ABS}
part lies in $S_{\gamma-1}$. Now we easily check that all the states
possibly present in $S_{\gamma-1}$ are mapped to $s_{2}$ by the
word $w_{\gamma}\dots w_{d}=b^{2}ab^{4n+m+7}$.

\subsection{From a Word to an Assignment.}

Since now we suppose that there is a reset word $w$ of length 
\[
d=12mn+8n-m+18.
\]
The following lemma is not hard to verify.

\begin{lem}
\leavevmode\label{lem: beginning of w}
\begin{enumerate}
\item \label{enu: unique merging path}Up to labeling there is a unique
pair of paths, both of a length $l\le d-2$, leading from $\mathit{cl}_{1}|\mathit{pipe}_{1}|s_{1}$
and $\mathit{cl}_{2}|\mathit{pipe}_{1}|s_{1}$ to a common end. They
are of length $d-2$ and meet in $s_{2}$.
\item \label{enu: start by aa}The word $w$ starts by $a^{2}$.
\end{enumerate}
\end{lem}
\begin{proof}
\leavevmode
\begin{enumerate}
\item The leading segments of both paths are similar since they stay within
the parts $\mathit{cl}_{1}$ and $\mathit{cl}_{2}$:
\begin{eqnarray*}
 & \mathit{pipe}_{1}|s_{1}\overset{a,b}{\longrightarrow}\dots\overset{a,b}{\longrightarrow}\mathit{pipe}_{1}|s_{12mn+4n-2m+6}\overset{a,b}{\longrightarrow}\mathit{abs}_{1}|\mathit{in}\overset{b}{\longrightarrow}\,\hspace{60bp}\\
 & \hspace{60bp}\overset{b}{\longrightarrow}\mathit{abs}_{1}|r_{1}\overset{b}{\longrightarrow}\mathit{abs}_{1}|\mathit{out}\overset{a}{\longrightarrow}\mathit{sp}_{1}\overset{b}{\longrightarrow}\dots\overset{b}{\longrightarrow}\mathit{sp}_{4n+6}.
\end{eqnarray*}
Once the paths leave the parts $\mathit{cl}_{1}$ and $\mathit{cl}_{2}$,
the shortest way to merge is the following:
\[
\begin{array}{lccclccc}
\mathit{cl}_{1}|\mathit{sp}_{4n+6} & \hspace{-6bp}\overset{b}{\longrightarrow}q_{1} & \hspace{-6bp}\overset{b}{\longrightarrow}q_{2} & \hspace{-6bp}\overset{b}{\longrightarrow}\,\,\,\,\,\dots & \overset{b}{\longrightarrow}q_{m-1} & \hspace{-6bp}\overset{b}{\longrightarrow}q_{m} & \hspace{-6bp}\overset{b}{\longrightarrow}s_{1} & \hspace{-6bp}\overset{b}{\longrightarrow}\\
\mathit{cl}_{2}|\mathit{sp}_{4n+6} & \hspace{-6bp}\overset{b}{\longrightarrow}q_{2} & \hspace{-6bp}\overset{b}{\longrightarrow}q_{3} & \hspace{-6bp}\overset{b}{\longrightarrow}\,\,\,\,\,\dots & \overset{b}{\longrightarrow}q_{m} & \hspace{-6bp}\overset{b}{\longrightarrow}s_{1} & \hspace{-6bp}\overset{b}{\longrightarrow}s_{2} & \hspace{-6bp}\overset{b}{\longrightarrow}
\end{array}\, s_{2}
\]
Having the description above it is easy to verify that the length
is $d-2$ and there is no way to make the paths shorter.
\item Suppose that $w_{1}w_{2}\neq a^{2}$. Any of the three possible values
of $w_{1}w_{2}$ implies that 
\[
\left\{ \mathit{cl}_{i}|\mathit{sp}_{3},\dots,\mathit{cl}_{i}|\mathit{sp}_{4n+6}\right\} \subseteq S_{2}
\]
 for each $i$. It cannot hold that $w=w_{1}w_{2}b^{d-2}$, because
in such case all $\mathit{cl}_{\dots}|\mathit{cca}|s_{b}$ states
would be active in any time $t\geq3$. So the word $w$ has a prefix
$w_{1}w_{2}b^{k}a$ for some $k\geq0$. If $k\leq4n+3$, it holds
that $\mathit{cl}_{i}|\mathit{sp}_{4n+6}\in S_{k+2}$ and therefore
$\mathit{cl}_{i}|\mathit{pipe}_{1}|s_{1}\in S_{k+3}$, which contradicts
the first claim. Let $k\geq4n+4$. Some state of a form $\mathit{cl}_{i}|\mathit{forcer}|q_{1,\dots}$
or $\mathit{cl}_{i}|\mathit{forcer}|r_{1,\dots}$ lies in $S_{k+2}$
for each $i$. This holds particularly for $i=1$ and $i=2$, but
there is no pair of paths of length at most 
\[
d-\left(4n+4\right)\geq d-k
\]
 leading from such two states to a common end.
\end{enumerate}
\end{proof}
The second claim implies that $\mathit{cl}_{i}|\mathit{pipe}_{1}|s_{1}\in S_{2}$
for each $i\in\left\{ 1,\dots,m\right\} $, so it follows that
\[
\delta\left(Q,w\right)=\left\{ s_{2}\right\} .
\]
Let us denote 
\[
\overline{d}=12mn+4n-2m+11
\]
 and 
\[
\overline{w}=w_{1}\dots w_{\overline{d}}.
\]
The following lemma holds because no edges labeled by $a$ are available
for final segments of the paths described in the first claim of Lemma
\ref{lem: beginning of w}.
\begin{lem}
\label{lem: end of w}~
\begin{enumerate}
\item \label{enu: final segment only b}The word $w$ can be written as
$w=\overline{w}b^{4n+m+7}$ for some word $\overline{w}$.
\item \label{enu: only sp}For any $t\geq\overline{d}$, no state from any
$\mathit{cl}_{\dots}$ part lie in $S_{t}$, except for the $\mathit{sp}_{\dots}$
states.
\end{enumerate}
\end{lem}
\begin{proof}
\leavevmode
\begin{enumerate}
\item Let us write $w=w_{1}w_{2}w'$. From Lemma \ref{lem: beginning of w}
it follows that
\[
\delta\left(\mathit{cl}_{1}|\mathit{pipe}_{1}|s_{1},w'\right)=\delta\left(\mathit{cl}_{2}|\mathit{pipe}_{1}|s_{1},w'\right)
\]
and $w'$ have to label some of the paths determined up to labeling
in Lemma \ref{lem: beginning of w}{\footnotesize{(1)}}. The final
$4n+m+7$ edges of the paths lead from $\mathit{cl}_{1}|\mathit{sp}_{1}$
and $\mathit{cl}_{2}|\mathit{sp}_{1}$ to $s_{2}$. All the transitions
used here are necessarily labeled by $b$.
\item The claim is easy to observe, since the first claim implies that $S_{t}$
is a subset of
\[
S'=\left\{ s\in Q\mid\left(\exists d\in\mathbb{N}\right)\delta\left(s,b^{d}\right)=s_{2}\right\} .
\]

\end{enumerate}
\end{proof}

The next lemma is based on properties of the parts $\mathit{cl}_{\dots}|\mathit{forcer}$
but to prove that no more $a$ follows the enforced factor $a^{2n+1}$
we also need to observe that each $\mathit{cl}_{\dots}|\mathit{cca}|out$
or each $\mathit{cl}_{\dots}|\mathit{cci}|\mathit{out}$\textbf{ }lies
in $S_{2n+4}$.
\begin{lem}
\label{lem: forcing plus b}The word $\overline{w}$ starts by $\overline{u}a^{2n+1}b$
for some $\overline{u}$ of length $2n+6$.\end{lem}
\begin{proof}
At first we prove that $\overline{w}$ starts by $\overline{u}a^{2n+1}$.
Lemma \ref{lem: beginning of w}{\footnotesize{(2)}} implies that
$\mathit{cl}_{1}|\mathit{pipe}_{2}|s_{1}\in S_{2}$, so obviously
some of the states $\mathit{cl}_{1}|\mathit{forcer}|q_{1,0}$ and
$\mathit{cl}_{1}|\mathit{forcer}|r_{1,0}$ lies in $S_{2n+6}$. If
$w_{2n+6+k}=b$ for some $k\in\left\{ 1,\dots,2n+1\right\} $, it
holds that $cl_{i}|\mathit{forcer}|q_{k,2}$ or $cl_{i}|\mathit{forcer}|r_{k,2}$
lies in $S_{2n+6+k}$. From such state no path of length at most $2n+3-k$
leads to $\mathit{cl}_{i}|\mathit{pipe}_{8}|s_{1}$ and therefore
no path of length at most \textbf{
\[
\left(2n+3-k\right)+\left(\alpha_{i}+6n-1\right)+\left(6n-2\right)+\beta_{i}+3=\overline{d}-\left(2n+6+k\right)
\]
}leads into $S'$, which contradicts Lemma \ref{lem: end of w}{\footnotesize{(2)}}.
It remains to show that there is $b$ after the prefix $\overline{u}a^{2n+1}$.
Lemma \ref{lem: beginning of w}{\footnotesize{(2)}} implies that
both $\mathit{cl}_{1}|\mathit{cca}|in$ and $\mathit{cl}_{1}|\mathit{cci}|\mathit{in}$\textbf{
}lie in $S_{2n+1}$, from which it is not hard to deduce that $\mathit{cl}_{1}|\mathit{cca}|out$
or $\mathit{cl}_{1}|\mathit{cci}|\mathit{out}$\textbf{ }lies in $S_{2n+4}$
and therefore $\mathit{cl}_{1}|q$ or $\mathit{cl}_{1}|r$ lies in
$S_{4n+7}$. Any path of length $\overline{d}-\left(4n+7\right)$
leading from $\mathit{cl}_{1}|q$ or $\mathit{cl}_{1}|r$ into $\overline{S}$
starts by an edge labeled by $b$.
\end{proof}
Now we are able to write the word $\overline{w}$ as
\[
\overline{w}=\overline{u}a^{2n+1}b\left(\overline{v}_{1}v'_{1}c_{1}\right)\dots\left(\overline{v}_{m}v'_{m}c_{m}\right)w_{\overline{d}-2}w_{\overline{d}-1}w_{\overline{d}},
\]
where $\left|\overline{v}_{k}\right|=6n-2$, $\left|v'_{k}\right|=6n-1$
and $\left|c_{k}\right|=1$ for each $k$ and denote $d_{i}=10n+\alpha_{i}+6$.
At time $2n+5$ the parts $\mathit{cl}_{\dots}|\mathit{pipe}_{6}$
and $\mathit{cl}_{\dots}|\mathit{pipe}_{7}$ record mutually inverse
sequences. Because there is the factor $a^{2n+1}$ after $\overline{u}$,
at time $d_{i}$ we find the information pushed to the first rows
of testers:

\begin{lem}
\label{prop: tester dichotomy}For each $i\in\left\{ 1,\dots,m\right\} $,
$j\in\left\{ 1,\dots,n\right\} $ it holds that\textup{
\begin{eqnarray*}
\mathit{cl}_{i}|\mathit{tester}|\mathit{level}_{x_{1}}|\left(1,x_{j}\right)\in S_{d_{i}} & \Leftrightarrow\\
\mathit{cl}_{i}|\mathit{tester}|\mathit{level}_{x_{1}}|\left(1,\neg x_{j}\right)\notin S_{d_{i}} & \Leftrightarrow & w_{2n-2j+2}\neq w_{2n-2j+3}.
\end{eqnarray*}
}\end{lem}
\begin{proof}
From the definition of \texttt{CCA} and \texttt{CCI} it follows that
at time $2n+5$ the parts $\mathit{pipe}_{6}$ and $\mathit{pipe}_{7}$
record the sequences $B_{\left(2n+3\right)}\dots B_{\left(2\right)}$
and $B'_{\left(2n+3\right)}\dots B'_{\left(2\right)}$ respectively,
where
\[
\begin{aligned}B_{\left(k\right)}= & \begin{cases}
\mathbf{1} & \mbox{if }w_{k}=w_{k+1}\\
\mathbf{0} & \mbox{otherwise}
\end{cases}\end{aligned}
\qquad\begin{aligned}B'_{\left(k\right)}= & \begin{cases}
\mathbf{0} & \mbox{if }w_{k}=w_{k+1}\\
\mathbf{1} & \mbox{otherwise}.
\end{cases}\end{aligned}
\]
Whatever the letter $w_{2n+6}$ is, Lemma \ref{lem: forcing plus b}
implies that 
\[
\mathit{cl}_{i}|x_{j}\in S_{4n+7}\Leftrightarrow\mathit{cl}_{i}|\neg x_{j}\notin S_{4n+7}\Leftrightarrow w_{2n-2j+2}\neq w_{2n-2j+3},
\]
from which the claim follows easily using Lemma \ref{lem: forcing plus b}
again.
\end{proof}
Let us define the assignment $\xi_{1},\dots,\xi_{n}\in\left\{ \mathbf{0},\mathbf{1}\right\} $.
By Lemma \ref{prop: tester dichotomy} the definition is correct and
does not depend on $i$: 
\[
\xi_{j}=\begin{cases}
\mathbf{1} & \mbox{if }\mathit{cl}_{i}|\mathit{tester}|\mathit{level}_{x_{1}}|\left(1,x_{j}\right)\notin S_{d_{i}}\\
\mathbf{0} & \mbox{if }\mathit{cl}_{i}|\mathit{tester}|\mathit{level}_{x_{1}}|\left(1,\neg x_{j}\right)\notin S_{d_{i}}.
\end{cases}
\]
The following lemma holds due to $\mathit{cl}_{\dots}|\mathit{limiter}$
parts.
\begin{lem}
\label{lem: limiting}For each $i\in\left\{ 1,\dots,m\right\} $ there
are at most two occurrences of $b$ in the word $v'_{i}$.\end{lem}
\begin{proof}
It is easy to see that $\mathit{cl}_{i}|\mathit{limiter}|s_{1,0}\in S_{10n+\alpha_{i}+6}$
and to note that
\[
v'_{i}=w_{10n+\alpha_{i}+7}\dots w_{16n+\alpha_{i}+5}.
\]
Within the part $\mathit{cl}_{i}|\mathit{limiter}$ no state except
for $s_{6n-2,0}$ can lie in $S_{16n+\alpha_{i}+5}$, because from
such states there is no path of length at most 
\[
\overline{d}-\left(16n+\alpha_{i}+5\right)=\beta_{i}+4
\]
 leading into $S'$. 

The shortest paths from $s_{1,0}$ to $s_{6n-2,0}$ have length $6n-3$
and each path from $s_{1,0}$ into $S'$ uses the state $s_{6n-2,0}$.
So there is a path $P$ leading from $s_{1,0}$ to $s_{6n-2,0}$ labeled
by a prefix of $v'$. We distinguish the following cases:
\begin{itemize}
\item If $P$ is of length $6n-3$, we just note that such path is unique
and labeled by $a^{6n-3}$. No $b$ occurs in $v'$ except for the
last two positions.
\item If $P$ is of length $6n-2$, it uses an edge of the form $s_{k,0}\overset{b}{\longrightarrow}s_{k+1,1}$.
Such edges preserve the distance to $s_{6n-2}$, so the rest of $P$
must be a shortest path from $s_{k+1,1}$ to $s_{6n-2,0}$. Such paths
are unique and labeled by $a^{6n-2-k}$. Any other $b$ can occur
only at the last position.
\item If $P$ is of length $6n-1$, it is labeled by whole $v'$. Because
any edge labeled by $b$ preserves or increases the distance to $s_{6n-2}$,
the path $P$ can use at most two of them.
\end{itemize}
\end{proof}
Now we choose any $i\in\left\{ 1,\dots,m\right\} $ and prove that
the assignment $\xi_{1},\dots,\xi_{n}$ satisfies the clause $\bigvee_{\lambda\in C_{i}}\lambda$.
Let $p\in\left\{ 0,1,2,3\right\} $ denote the number of unsatisfied
literals in $C_{i}$. 

As we claimed before, all tester columns corresponding to any $\lambda\in L_{\phi}$
have to be deactivated earlier than other columns. Namely, if $\mathit{cl}_{i}|\mathit{tester}|\mathit{level}_{x_{1}}|\left(1,\lambda\right)$
is active at time $d_{i}$, which happens if and only if $\lambda$
is not satisfied by $\xi_{1},\dots,\xi_{n}$, the word $v'_{i}c_{i}$
must not map it to $\mathit{cl}_{i}|\mathit{pipes}_{3}|s_{1,\mu\left(\lambda\right)}$.
If $\mathit{cl}_{i}|\mathit{tester}|\mathit{level}_{\lambda}$ is
of type \texttt{INC($\lambda$)}, the only way to ensure this is to
use the letter $b$ when the border of inactive area lies at the first
row of $\mathit{cl}_{i}|\mathit{tester}|\mathit{level}_{\lambda}$.
Thus each unsatisfied $\lambda\in C_{i}$ implies an occurrence of
$b$ in corresponding segment of $v'_{i}$: 
\begin{lem}
\label{lem: at least p of b}There are at least $p$ occurrences of
the letter $b$ in the word $v'_{i}$.\end{lem}
\begin{proof}
Let $\lambda_{1},\dots,\lambda_{p}$ be the unsatisfied literals of
$C_{i}$. From Lemma \ref{prop: tester dichotomy} it follows easily
that
\[
\mathit{cl}_{i}|\mathit{tester}|\mathit{level}_{\lambda_{k}}|\left(1,\lambda_{k}\right)\in S_{d_{i}+3\kappa\left(\lambda_{k}\right)}
\]
for each $k\in\left\{ 1,\dots,p\right\} $. The part $\mathit{cl}_{i}|\mathit{tester}|\mathit{level}_{\lambda_{k}}$
is of type \texttt{INC($\lambda_{k}$)}, which implies that any path
of the length 
\[
\left(\overline{d}-3\right)-\left(d_{i}+3\kappa\!\left(\lambda_{k}\right)\right)
\]
starting by $a$ takes $\mathit{cl}_{i}|\mathit{tester}|\mathit{level}_{\lambda_{k}}|\left(1,\lambda_{k}\right)$
to the state $\mathit{cl}_{i}|\overline{\lambda}$, which lies outside
$S_{\overline{d}-3}$, as it is implied by Lemma \ref{lem: end of w}{\footnotesize{(2).}}
We deduce that $w_{d_{i}+3\kappa\left(\lambda_{k}\right)+1}=b$.
\end{proof}
By Lemma \ref{lem: limiting} there are at most two occurrences of
$b$ in $v'_{i}$, so we get $p\leq2$ and there is at least one satisfied
literal in $C_{i}$.

\section*{References}

\bibliographystyle{elsarticle-harv}
\bibliography{C:/Users/Vojta/Desktop/SYNCHRO2/bib/ruco}

\end{document}